\documentclass[12pt]{amsart}
\usepackage[margin=1in]{geometry}
\usepackage[english]{babel}
\usepackage[utf8]{inputenc}
\usepackage[all]{xy}
\usepackage[labelfont=bf,labelsep=period,justification=raggedright]{caption}
\usepackage[colorinlistoftodos]{todonotes}
\usetikzlibrary{calc}
\newcommand{\Chi}{\mathrm{X}}

\usepackage[mathscr]{eucal}
\usepackage{amsaddr,multirow, amsmath,amssymb,amsthm,asymptote,wrapfig,dsfont,amsfonts,hyperref,fancyhdr,listings,csquotes,listings,xurl,subcaption,amsfonts,mathrsfs,hyperref,tkz-fct,tikz,multicol,bm,svg,textcomp,graphicx, color,indent first,hyperref,tkz-fct,tikz,xcolor,soul,bbold,stmaryrd,algorithm,algpseudocode,tkz-fct,tikz,tikz-cd,tabularx,booktabs,float,enumitem}

\usepackage[colorinlistoftodos]{todonotes}
\newcommand{\mycomment}[1]{}

\PassOptionsToPackage{dvipsnames,svgnames}{xcolor}

\theoremstyle{plain}
\newtheorem{thm}{Theorem}

\theoremstyle{definition}
\newtheorem{defn}[thm]{Definition}

\theoremstyle{remark}

\numberwithin{equation}{section}
\numberwithin{thm}{section}

\let\oldtocsection=\tocsection
\let\oldtocsubsection=\tocsubsection
\let\oldtocsubsubsection=\tocsubsubsection

\renewcommand{\tocsection}[2]{\hspace{0em}\oldtocsection{#1}{#2}}
\renewcommand{\tocsubsection}[2]{\hspace{1em}\oldtocsubsection{#1}{#2}}
\renewcommand{\tocsubsubsection}[2]{\hspace{2em}\oldtocsubsubsection{#1}{#2}}

\title[Enhancing Ethereum's Security with LUMEN]{Enhancing Ethereum's Security with LUMEN, a Novel Zero-Knowledge Protocol Generating Transparent and Efficient zk-SNARKs}
\author[Yunjia Quan]{Yunjia Quan}

\usepackage[backend=biber,style=ieee,sorting=ynt]{biblatex}
\begin{document}

\maketitle
\thispagestyle{empty}
\begin{abstract}

This paper proposes a novel recursive polynomial commitment scheme (PCS) and a new polynomial interactive oracle proof (PIOP) protocol, which compile into efficient and transparent zk-SNARKs (zero-knowledge succinct non-interactive arguments of knowledge). The Ethereum blockchain utilizes zero-knowledge Rollups (ZKR) to improve its scalability (the ability to handle a large number of transactions), and ZKR uses zk-SNARKs to validate transactions. The currently used zk-SNARKs rely on a trusted setup ceremony, where a group of participants uses secret information about transactions to generate the public parameters necessary to verify the zk-SNARKs. This introduces a security risk into Ethereum's system. Thus, researchers have been developing transparent zk-SNARKs (which do not require a trusted setup), but those are not as efficient as non-transparent zk-SNARKs, so ZKRs do not use them. In this research, I developed LUMEN, a set of novel algorithms that generate transparent zk-SNARKs that improve Ethereum's security without sacrificing its efficiency. Various techniques were creatively incorporated into LUMEN, including groups with hidden orders, Lagrange basis polynomials, and an amortization strategy. I wrote mathematical proofs for LUMEN that convey its completeness, soundness and zero-knowledgeness, and implemented LUMEN by writing around $8000$ lines of Rust and Python code, which conveyed the practicality of LUMEN. Moreover, my implementation revealed the efficiency of LUMEN (measured in proof size, proof computation time, and verification time), which surpasses the efficiency of existing transparent zk-SNARKs and is on par with that of non-transparent zk-SNARKs. Therefore, LUMEN is a promising solution to improve Ethereum's security while maintaining its efficiency.
\end{abstract}

\newpage
\thispagestyle{empty}
\tableofcontents

\pagebreak
\setcounter{page}{1}

\section{Introduction}\label{sec:introduction}

The Ethereum blockchain primarily uses a consensus mechanism called Nakamoto, which utilizes a proof-of-work algorithm, and Ethereum uses layer $2$ scaling solutions to process more transactions efficiently. The best layer $2$ scaling solution that Ethereum currently uses is zero-knowledge Rollup (ZKR). ZKR processes thousands of Ethereum transactions in a batch off-chain and posts some summary data to the Ethereum blockchain, as well as a zero-knowledge proof of those data to verify the validity of those transactions and data proposed by the ZKR. 

Zero-knowledge proofs (ZKPs) are mathematical proofs that allow parties to prove their knowledge of a statement without revealing the statement, and ZKR uses ZKPs to validate off-chain transactions. The ZKP that ZKR uses comes in the form of zk-SNARKs (zero-knowledge succinct arguments of knowledge). The current zk-SNARKs used by ZKR rely on a trusted setup ceremony, where a group of participants use secret information about transactions to generate the public parameters and information necessary to verify the zk-SNARKs. However, there are security concerns surrounding the trusted setup because information can be leaked during the process and the zk-SNARKs' security would be compromised. 

Thus, researchers have been looking to develop transparent zk-SNARKs: zk-SNARKs that do not use a trusted setup ceremony. Instead, they rely on publicly verifiable randomness to set up the public parameters. Thus, transparent zk-SNARKs are much more secure than non-transparent ones. However, the existing transparent zk-SNARKs are much less efficient than the non-transparent zk-SNARKs due to their larger proof sizes and long computation time, so ZKR does not want to deploy the transparent zk-SNARKs. In this research, I aim to craft transparent zk-SNARKs that are as efficient as non-transparent zk-SNARKs so that I can improve Ethereum's security without sacrificing its efficiency.  

\subsection{Related Works} \label{relatedstudies}

The first zk-SNARKs were introduced in 2013 [9] as efficient zero-knowledge proofs to provide succinct and computationally sound arguments. The small proof size and constant verification time gave zk-SNARKs an advantage over other zero-knowledge systems. However, a significant drawback is the reliance of zk-SNARKs on a trusted setup, which introduces security risks to ZKR systems. 

With efforts in removing the trusted setup, there are zk-SNARKs employing a universal setup (a one-time setup that is thrown away after usage), which often utilizes an updatable structure reference string (SRS) and the technique of algebraic holographic proofs (AHP). Examples of universal setups include MIRAGE [25], LUNAR [16] based on PHPs (Polynomial Holographic IOPs), MARLIN based on AHPs [18], BlockMaze [23], and Bitansky's work based on PCPs (probabilistic checkable proofs) [10]. 

There is also a family of universal zk-SNARKs called PLONK, which is designed to offer a more efficient and simplified construction of universal and updatable zk-SNARKs compatible with a recursive proof composition [21]. Examples within this family are Sonic [26], HyperPlonk [17], Plookup [20], and Ambrona's optimization of Plonk [6]. 

Researchers have also previously developed transparent SNARKs, but their efficiency is much lower than the non-transparent SNARKs. Thus, they are not able to be incorporated by ZKP and Ethereum. Notable among them are Bulletproofs, which utilize a discrete log-based approach [12]. However, the verification time of the bulletproofs is linear with the size of the polynomial, which is too long. Another protocol based on discrete logarithm is Hyrax [1], which shares a similar problem. 

Transparent SNARKs generated through various techniques such as Diophantine Argument of Knowledge [14], list polynomial commitments [24], lattice-based schemes [3], and univariate IOP [8] also face challenges such as large proof sizes, multiple trust assumptions, or extended verification times. On the other hand, scalable, transparent arguments of knowledge of zero knowledge, known as zk-STARKs, are not succinct and have significantly larger proof sizes than zk-SNARKs [7]. 

Efforts have also been directed towards recursive proof composition strategies for SNARKs to reduce lengthy computations. Some techniques have been reported in the literature such as Pointproofs [22], Cramer's work on amortized complexity [19], Halo [11], Botrel's work on faster montgomery multiplication [4], and Halo Infinite [2]. However, these approaches often rely on creating provers for the entire SNARK circuit, which leads to considerable verification time. 

The existing transparent zk-SNARKs and recursive approaches are all facing one significant challenge: they are too inefficient. In this research, I aim to address this challenge. I aim to design transparent zk-SNARKs that are secure and efficient. Compared to the currently used non-transparent zk-SNARKs, I aim to design transparent zk-SNARKs that enhance Ethereum's security without sacrificing its efficiency. 

\subsection{My Contributions}

I constructed a new PCS with recursive proof composition and a witness-extended emulation proof, a novel PIOP protocol, mathematical proofs of completeness, soundness, and zero-knowledge, and implemented my algorithms - which are called LUMEN. LUMEN's recursive PCS serve as a cryptographic primitive that allows a prover to commit to a polynomial so that they can later reveal the value of the polynomial at specific points without revealing the entire polynomial, and LUMEN's PIOP used with its PCS enables the creation of transparent zero-knowledge proofs. 

My unique designs within LUMEN contribute to its strong security and efficiency. In my new PCS, I creatively used groups with hidden orders, inventive encoding vectors, and masking polynomials to ensure that it is inherently transparent. In my recursive proof composition method, I designed a new method with an amoritzation strategy and a hash function, which is able to significantly reduce the verification time by not having to verify each individual ZKP. Moreover, a new PIOP protocol was crafted that incorporates Lagrange basis polynomials, sets of matrices, and auxiliary polynomials. Those new and creative aspects of LUMEN altogether allow LUMEN to surpass existing transparent zk-SNARKs in efficiency. 

Throughout the creation of LUMEN, I ensured that no trusted setup would be used and that I can mathematically prove the completeness, soundness, and unforgeability of the transparent zk-SNARKs generated by LUMEN. Furthermore, I ensured that the polynomials that I craft in the PCS and PIOP are going to generate efficient and secure arguments of knowledge. Moreover, my code implementation of LUMEN reveals its practicality and is novel in itself for transforming algorithms and equations into code. 

With its unique design of the PCS and PIOP, LUMEN is a promising solution for Ethereum to implement to improve its security without sacrificing its efficiency. This has significant implications for Ethereum because an increase in Ethereum security has the potential to impact the security standards of the entire cryptocurrency market. In addition, others can modify and adapt my algorithms to develop algorithms for other cryptocurrencies. 

\subsection{Paper organization}

The remainder of this paper is organized as follows. Section \ref{2} introduces the preliminaries of cryptography and algebra used in LUMEN's PCS or PIOP. Section \ref{3} reveals the framework of LUMEN's algorithms and their relationship to generate transparent zk-SNARKs. Section \ref{4} covers the construction of LUMEN's PCS, outlining its central argument and recursive proof composition. Moreover, I prove that LUMEN's PCS has witnessed extended emulation. Section \ref{5} contains the details of the PIOP protocol. Section \ref{6} presents the security analysis of LUMEN and mathematically proves the completeness, soundness, and zero-knowledge of the arguments of knowledge compiled by LUMEN. Section \ref{7} presents the implementation and results of LUMEN and compares it with zk-SNARKs. Finally, Section \ref{8} concludes this research by analyzing the results and implications of LUMEN. 

\section{Preliminaries}\label{2}

Polynomial Commitment Schemes are algorithms that generate a commitment to a polynomial, which can later be used by a verifier to confirm the evaluations of the commitment polynomial. PCS is a commonly used part of zk-SNARKs that locks in the prover's claims about their knowledge without revealing it. I first define a general commitment scheme: 

\begin{defn}
(Commitment Scheme) [15]. A commitment scheme $\Gamma$ is a tuple of Probabilistic Polynomial Time (PPT) algorithms $\Gamma = (\textbf{Setup, Commit, Open})$ where
\begin{itemize}
    \item $\textbf{Setup}(1^{\gamma}) \rightarrow \mathrm{pp}$ generates the public parameters $\mathrm{pp}$
    \item $\textbf{Commit}(pp;x) \rightarrow (C;r)$ takes a secret message $x$ and outputs a public commitment $C$
    \item $\textbf{Open}(pp,C,x,r) \rightarrow b \in \{0,1\}$ takes message $x$ and opens hint $r$ to verify the opening of commitment $C$
\end{itemize}
\end{defn}

Then, I define a polynomial commitment scheme as follows: 

\begin{defn}
(Polynomial Commitment Scheme) [13] A polynomial commitment scheme (PCS) is a tuple of the Polynomial Probability Time (PPT) algorithms \\ $(\textbf{Setup, Commit, Open, Eval})$ such that $(\textbf{Setup, Commit, Open})$ is a binding commitment scheme and the evaluation algorithm is an interactive public-coin protocol between the PPT prover $\mathcal{P}$ and the verifier $\mathcal{V}$. 
\end{defn}

Alongside my novel PCS, I used a polynomial interactive oracle proof (PIOP) protocol. They are useful for making the zk-SNARK in LUMEN's PCS transparent and efficient. I define an interactive oracle proof as follows: 

\begin{defn}(Interactive Oracle Proof) [24]
An interactive proof of knowledge for an NP relation $\mathcal{R}$ is an interactive protocol between a prover $\mathcal{P}$ and a verifier $\mathcal{V}$, where $\mathcal{P}$ has a private input $x$ and both parties have a common public input $y$ such that $(y, x) \in \mathcal{R}$. 
\end{defn}

A polynomial interactive oracle proof follows the definition of an IOP, as PIOP is a specific case of an IOP. 

\vspace{2mm}

Now, I introduce the notions of vanishing and Lagrange polynomials, which will be later used in LUMEN's algorithms. Consider a multiplicative subgroup $\mathbb{H} \in \mathbb{F}$ for a field $\mathbb{F}$. A vanishing polynomial approaches zero at any point. I define the vanishing polynomial as $$z_s(x) = x^{\lvert \mathbb{H} \rvert} - 1$$

I denote an $n^{th}$ Lagrange Polynomial by $\Delta_n(x)$ over a set of points $x_i$ where $1 \leq i \leq k \in \mathbb{Z}$ such that 
\begin{align*}
\Delta_{n}(x) = \prod_{0 \leq m < k \& m \neq n} \frac{x - x_m}{x_n - x_m} =\frac{n}{\lvert \mathbb{H} \rvert} \cdot \frac{x^{\lvert \mathbb{H} \rvert} - 1}{x - n}
\end{align*}

The second definition is typically used to simplify equations. Furthermore, I define a bivariate Lagrange polynomial as follows:

\begin{defn}(Bivariate Lagrange Polynomial). Let $\mathbb{F}$ be a finite field, and $\mathbb{H}$ be a multiplicative subgroup of $\mathbb{F}$. The bivariate Lagrange polynomial can be expressed as follows: 
$$\Lambda_{\mathbb{H}}(X, Y):=\frac{Z_{\mathbb{H}}(X) \cdot Y-X \cdot \mathcal{Z}_{\mathbb{H}}(Y)}{n \cdot(X-Y)}$$
\end{defn}

The concise form of the Lagrange polynomial helps reduce the verification time in the LUMEN PIOP. I prove the following theorem which is useful for the construction of LUMEN's PIOP. 

\begin{thm}
Let $\mathbb{F}$ be a finite field and $\mathbb{H}$ be a multiplicative subgroup of $\mathbb{F}$. Then, I obtain $\Lambda_{\mathbb{H}}(X, Y)= \sum_{h \in \mathbb{H}} \Delta_{h}(X) \cdot \Delta_{h}(Y)$ for the Lagrange polynomial $\Delta$. 
\end{thm}

\begin{proof}
I use the following algebraic manipulations to show this theorem: 
$$
\begin{aligned}
\sum_{h \in \mathbb{H}} & \Delta_{h}(X) \cdot \Delta_{h}(Y) =\sum_{h \in \mathbb{H}} \frac{h^{2}}{n^{2}} \frac{\mathcal{Z}_{\mathbb{H}}(X) \cdot \mathcal{Z}_{\mathbb{H}}(Y)}{(X-h)(Y-h)}=\frac{\mathcal{Z}_{\mathbb{H}}(X) \cdot \mathcal{Z}_{\mathbb{H}}(Y)}{X-Y} \sum_{h \in \mathbb{H}} \frac{h^{2}}{n^{2}} \frac{X-Y}{(X-h)(Y-h)} \\
& =\frac{\mathcal{Z}_{\mathbb{H}}(X) \cdot \mathcal{Z}_{\mathbb{H}}(Y)}{n \cdot(X-Y)} \sum_{h \in \mathbb{H}} \frac{h^{2}}{n}(\frac{X-h}{(X-h)(Y-h)}+\frac{X-Y-(X-h)}{(X-h)(Y-h)} \\
& =\frac{\mathcal{Z}_{\mathbb{H}}(X) \cdot \mathcal{Z}_{\mathbb{H}}(Y)}{n \cdot(X-Y)} \sum_{h \in \mathbb{H}} \frac{h^{2}}{n}\left(\frac{1}{Y-h}-\frac{1}{X-h}\right) \\
&=\frac{1}{n \cdot(X-Y)}(\boldsymbol{z}_{\mathbb{H}}(X) \sum_{h \in \mathbb{H}} h \cdot \Delta_{h}(Y)-\mathcal{Z}_{\mathbb{H}}(Y) \sum_{h \in \mathbb{H}} h \cdot \Delta_{h}^{\mathbb{H}}(X)) \\
& =\frac{\left(\mathcal{Z}_{\mathbb{H}}(X) \cdot Y-\mathcal{Z}_{\mathbb{H}}(Y) \cdot X\right)}{n \cdot(X-Y)} \\
& = \Lambda_{\mathbb{H}}(X, Y)
\end{aligned}
$$
\end{proof}

\section{Framework} \label{3}

LUMEN consists of a PCS and a polynomial interactive oracle proof (PIOP) protocol. When I apply the Fiat-Shamir heuristic to LUMEN, I acquire transparent zk-SNARKs: transparent zero-knolwedge succinct non-interactive arguments of knowledge. The figure outlining the framework of LUMEN can be found on the next page. 

\begin{figure}
    \centering
    \includegraphics[width=0.96\textwidth]{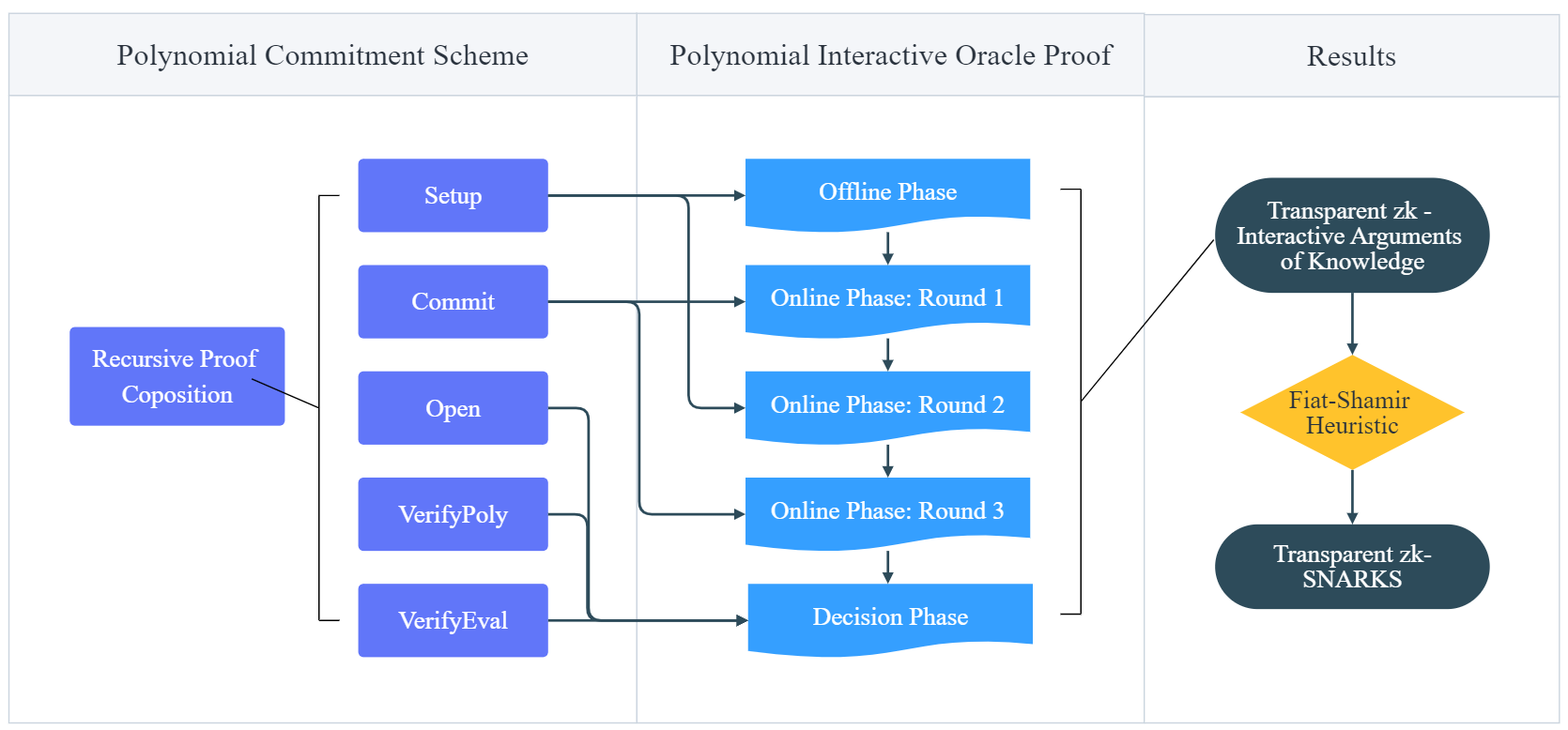}
    \caption{Framekwork of LUMEN \\ \textbf{Source}: Myself}
    \label{fig:framework}
\end{figure}

\section{PCS Construction}\label{4}

\subsection{Polynomial Commitment Scheme}

I introduce the construction of LUMEN's polynomial commitment scheme, which is used with a recursive proof composition method to altogether generate aggregated polynomial commitments. LUMEN's PCS consists of five main algorithms: (\textbf{Setup}, \textbf{Commit}, \textbf{Open}, \textbf{VerifyPoly}, and \textbf{VerifyEval}). 

\noindent \textbf{Setup}: 
\begin{itemize}
    \item Sample a group of unknown order $G$ and a vector $v \in \mathbb{Z}^{\alpha}$; randomly choose $g \in G$ and let $u = \{g,g^2,\dots,g^{\alpha}\}$
    \item Choose a function $f(x) \in \mathbb{Z}/(x^d - 1)$
    \item Generate a masking polynomial $p_1(x)$ (monic) and a witness polynomial $p_2(x)$ with degrees $< d$
    \item Publish the public parameters $(G,v,u,d,p_2(x))$
\end{itemize}
\textbf{Commit}: 
\begin{itemize}
    \item Given the public parameters, compute $$q(x) = \sum_{i \in \lvert v \rvert} v_i p_1(x) + p_2(x)^g$$
    \item Compute the following: 
    $$c(x) = \frac{\sum_{g \in u}g \cdot f(x) + p_1(x)}{\epsilon \cdot q(x)}$$
    \item Send the commitment $c(x)$ along with $q(x)$
\end{itemize}
\textbf{Open}: 
\begin{itemize}
    \item Verify that  
    $$\epsilon \cdot q(x) \cdot c(x) \equiv p_2(x) \mod x^d - 1 $$
\end{itemize}

\noindent \textbf{VerifyPoly}: 

In this algorithm, I wish to validate that $Open(c(x)$,$q(x)$,$p_2(x)) = 1$ for a true commitment $c(x)$. Consider the vectors $u,v$ that encode the properties of $c(x)$ and the vector $w$ defined as follows: 
$$u_i \cdot v_i = w_i$$

Let $r(x),s(x),t(x)$ be three vectors of polynomials that encode the properties of $c(x)$ along with $u,v,w$. Then, I need the following constraint to be satisfied: 
\begin{align*}
\sum_{i = 1}^{\alpha} u_i \cdot r_i(x) &+ \sum_{i = 1}^{\alpha}v_i \cdot s_i(x) + \sum_{i = 1}^{\alpha} w_i(x^{\alpha}t_i(x)-(x^i+x^{-i})) +\sum_{i=1}^{\alpha} \frac{u_{i} \cdot t_{i}(x)}{v_{r} \cdot x^{\alpha}}=0
\end{align*}

where I define the functions $r,s,t$ as follows:
\begin{align*}
r_{i}(x)=\prod_{i = 1}^{d} f\left(u_{i}\right) \cdot x^{i} \text{ }; \text{ } s_{i}(x)=\prod_{i = 1}^{d} p_1(v_i) \cdot x^i \text{ }; \text{ } t_{i}(x)=& \prod_{i = 1}^{d} p_2(v_i) \cdot x^i
\end{align*}
However, the prover cannot commit directly to $r,s,t$, I define the following functions: 
\begin{align*}
a(x, y)&=\sum_{i=1}^{\alpha} u_{i}\left(x\right) \cdot x^{i} y^{-i}+y^{\alpha}\sum v_i(x^i) \cdot x^{-i} - x^{\alpha} \sum w_{i}(Y) \cdot Y\\
b(x, y) &=x^{\alpha} \sum u_{i}\left(y^{i}\right)-\sum v_{i}(x)\left(y^{-\alpha}-x^{i}\right)-\sum w_{i}\left(y^{i}\right)\left(x^{-i}-x^{i}\right) \\
d(x, y) &=\sum\left(x^{-i}-y^{i}\right) u_{i} \cdot r_{i}-x^{-\alpha} \sum v_{i} \cdot w_{i}+ \prod_{i = 1}^{\alpha} w_{i} \cdot x^{-i-\alpha} \\
e(x, y) &= a(x, 1)+b(1, y)-d(x, y)
\end{align*}

With these equations, the prover can commit to $a(x,y)$ due to $a(x,y) = a(xy,1)$ and demonstrate that the commitment $Commit(u,v,p_1,p_2,f)$ is valid. The proof discloses the commitment to $m = a(1,q)$, $n = b(p,1)$, and $r = e(p,q)$, then the verifier can easily verify the commitment as long as the following equation is satisfied: 
$$r = p(d(m,n)+n) - q$$

\noindent \textbf{VerifyEval}: The goal of the evaluation is to convince the verifier that $c(x)$ is a commmitment and that $f(x_1) = y_1$ for a $x_1 \in u$ and $f(x_2) = y_2$ for a $x_2 \in v$. 

\vspace{7mm}

\vspace{-5mm}

\begin{enumerate}[after=\vspace{-\baselineskip},before=\vspace{-0.5\baselineskip}]
    \item The verifier $\mathcal{V}$ picks a bit $b$ randomly and computes $\hat{b} \equiv b \mod d$
    \item The prover $\mathcal{P}$ computes $h_1(x)$ and $h_2(x)$ such that $f = h_1(x) \cdot b - h_2(x)$
    \item $\mathcal{P}$ computes $v^{\prime} = p_1(v) + \alpha \cdot v$ and $u^{\prime} = p_2(u) + \alpha \cdot u$
    \item $\mathcal{P}$ computes the following: 
    $$c^{\prime} = \frac{\sum_{i \in \vert u \vert} u_i f(x) + p_1(v)}{\alpha \cdot q(v)}$$
    \item $\mathcal{P}$ sends $\hat{b},v^{\prime},u^{\prime},c^{\prime}$ to $\mathcal{V}$
    \item $\mathcal{V}$ checks that $$f(v) = \hat{b} \cdot v^{\prime} - u^{\prime} \cdot \lvert v \rvert$$
    $\mathcal{V}$ returns $0$ if the condition is not met, $1$ if it is. 
\end{enumerate}

\vspace{1mm} \subsection{Recursive Protocol}
Using LUMEN's recursive proof composition, I wish to generate aggregated knowledge arguments because they are more efficient for ZKRs. In the recursion protocol, I use an amortization strategy such that I skip the proof verification at each recursion level to reduce the verification time. I also use a hash function during the regression to reduce the size of the final arguments and increase its security. The full protocol of the recursive proof composition is as follows:

$$
R=\left\{\begin{array}{l}
\left(H(G_{0}), f_{0}(x),\left(p_{1}(x)\right)_{0},\left(p_{2}(x)\right)_{0}, k \in[1, \alpha]\right. \wedge \left(u,v, w\right) \text { and }(r,s,t) \text { where} \\
\sum_{i=1}^{\alpha} u_{i} \cdot r_{i}(x)+\sum_{i=1}^{\alpha} v_{i} \cdot s_{i}(x) +\sum_{i=1}^{\alpha} w_{i} \cdot t_{i}(x)=g^{k} \\
\wedge G_{t}= \sum_{i = 0}^{n} H(\operatorname{Agg}\left(G_0, c_0, g\left(x,\left(p_{1}(x)\right)_{0}, p_{2}(x)\right)_{0}\right)) \\
\wedge C_{t}(x) = c_n(H(c_{n-1}(x))) + c_{n-1}(H(c_{n-2}(x))) + \cdots
\end{array}\right. 
$$

where $H(x)$ is the $KECCAK-256$ hash function, the standard hash function used by Ethereum. 

I recursively composed the hidden-order groups and related parameters and did not verify the individual arguments of knowledge. Instead, I continue to aggregate the arguments and verify the values of $G_t$ and $C_t(x)$ at the end to reduce the verification time while maintaining security. 

\subsection{Witness-Extended Emulation}

\begin{thm}
THe proposed PCS has a witness-extended emulation, which allows for easy ZKR implementation.  
\end{thm}

\begin{proof}
Consider simulator $S$ which acts as an honest prover: Let $S$ generate $h \in G$ and $H = \{h,h^2,\dots,h^{\alpha} \}$. I also introduce an extractor $X$ and adversary $A$ that can generate accepted transcripts. This leads to the following challenges for the prover: $x,y,y_1,z_1,z_2,\dots,z_7,\\ v_1,v_2,\dots,v_{10}$. For the prover, I have
\begin{align*}
t(x,y) &=\frac{\left(v_{1}+z_{1} v_{10}+z_{1}^{2} v_{9}+z_{1}^{3} v_{8}+z_{1}^{4} v_{4}+z_{1}^{5} v_{3}\right) - v{11}}{z_{6}-x} \\
&+z_{4} \frac{\left(v_{8}+z_{2} v_{4}\right) - v_3}{z_{4}-y}+z_{2}^{3} \frac{v_{12}-v_{9}}{z_{3}-x y}+z_{2}^{3} \frac{v_{13}-v_{2}}{z_{3}-y_{o}}+z_{2}^{4} \frac{v_{13}-v_{4}}{z_{3}-y_{1}}
\end{align*}

From $t(x,y)$, I substitute in $r_1(x),r_2(x),r_3(x), \\ r_4(x),c^{\prime}(x)$, 
\begin{align*}
t(x,y) &=\frac{r_1(x)-\left(v_{10}+z_2v_7 + z_2^{2}v_4+z_2^{3}v_1\right)}{x-y^{-1}} \\
&+z_{3} \frac{r_2(x)-\left(v_{8}+z_{1} v_{3}\right)}{x-y}+z_{3}^{2} \frac{r_3(X)-v_{3}}{x-y^2}+z_{3}^{3} \frac{v_{10}-c^{\prime}(x)}{x-y_{o}}+z_{3}^{4} \frac{r_4(x)-v_{12}}{X-y_{1}}
\end{align*}

I run the prover again with the challenge of extracting $t(x,y)$, and obtain the following equations:  
\begin{align*}
&r_1(x)=v_{1}+z_{1} v_{10}+z_{1}^{2} v_{9}+z_{1}^{3} v_{8}+z_{1}^{4} v_{4}+z_{1}^{5} v_{3} \\
&r_2(y)=v_{8}+z_{2} v_{4} + z_{2}^2 v_{5} \text{ ; } r_3(x y)=v_3 - v_{10} \cdot v_7
\end{align*}

Moreover, I get that
\begin{align*}
r_4(y) = \frac{v_5 \cdot v_8}{v_1} + v_9 \text{ ; } r_4(y_0) = v_2 - v_4 \text{ ; } r_4(y_1) = (v_1 - v_7)v_2
\end{align*}

Now that I have defined all of the polynomials necessary, I seek to find the relationships between the polynomials and commitments. I used the extractor as a verifer to verify the commitments and equations that the polynomials would have to satisfy. After the computations, I obtain:
\begin{align*}
&c^{\prime}\left(y_{o}\right)=v_{7}, r_3(x)= v_{2}, r_2(x)=v_{1}, 
\\
&c^{\prime}\left(y_{1}\right)=v_{9},
r_4(x)= v_{3}, 
r_4(y)=v_{5}, \\
&c^{\prime}(y)=v_{6}, 
r_1(y)=v_{4},
r_2(x y)=v_{8}
\end{align*}

which establishes that 
$$c(x)= \mathrm{Commit} \left(v,z, r_1(x), r_2(x),r_3(x)\right)$$

Finally, with the challenge of outputting the correct commitment, I obtain the polynomial $t^{\prime}(x,y)$: 
\begin{align*}
t^{\prime}(x,y) &= (r_3(x)-r_1(x))x^{2\alpha} \cdot r_1(xy)+ r_2(y) \cdot \sum_{i=1}^{\alpha} (x^i - x^{-i})y^{\alpha} - r_4(xy)
\end{align*}

From this, I find that $t(x,y) = t^{\prime}(xy,1)$, I can conclude that the constraint in the protocol is satisfied, and my PCS indeed demonstrates witness-extended emulation.
\end{proof}

\section{PIOP Protocol}\label{5}

LUMEN's PIOP consist of five phases: an offline phase, three rounds of an online phase, and a decision phase. These phases go from setting up the public parameters to generating the transparent interactive arguments of knowledge and ultimately proving the validity of those arguments of knowledge. 

\subsection{Offline Phase}

The offline phase is preparatory and typically involves setting up the public parameters used in the PIOP protocol. I first define a universal relation that will be central to the relation encoder to craft relationships and generate the zero-knowledge proof on. 

Let $\mathbb{F}$ be a finite field, and let $k$ be a positive integer. I seek to define a universal relationship on the tuple as follows:
$$(\mathbb{F},k,M_1,M_2,h^{\prime}(x), h^{\prime}(x), h^{\prime \prime}(x), m, n)$$

where $M_{1}, M_{2} \in F^{n \times n}$, $max\left\{\left\|M_{1} \mid,\right\| M_{2} \|\right\} \leqslant k$, $h(x), h^{\prime}(x), h^{\prime \prime}(x) \in F_{\alpha-1}[x]$, and $\mathbb{K}=\left\{h^{\prime \prime}(1), \ldots, h^{\prime \prime}(m)\right\}$ such that
\begin{align*}
&f(x)=\sum_{k \in \mathbb{K}} h(k), Q(X)+h^{\prime}(x) \cdot n \\
&h^{\prime \prime}(x)=1-Q(x) \cdot h(x)+k h^{\prime}(x)\\
&Q(x)=\prod_{k \in K}\left(P_{1}(k+1)-P_{2}\left(k^{2}-x\right)\right)
\end{align*}

Then, the relation is as follows: 
$$M_{1} \cdot h^{\prime}(x)-h \cdot h^{\prime \prime}(x)+\sum_{k \in K}\left(Q(x)-h(k) \cdot M_{2}\right) \cdot h^{\prime}(k)=0$$ 

The relation encoder $\mathcal{R}$ takes in $\left(F, A, B, M_{1}, M_{2}, M_{1}\right)$ and generates eight polynomials: $P_{1}, P_{2}, P_{3}, P_{4}, Q_{1}, Q_{2}, Q_{3}, Q_{4} \in F_{m}[x]$. 

I define the $Q_i$'s as follows: 
\begin{align*}
Q_{1}(x)&=Q(x-\alpha)+\sum_{k \in K} M_{1}^{\text {row }(k)\text{col }(k)} T P_{1}(x, 1) +n \cdot P_{3}(x, \alpha)\\
Q_{2}(X)&=Q(x+\alpha)+\sum_{k \in K} M_{2}^{\text {row }(k) \text { col }(k)}\left(P_{2}(x, 1)-Q_{1}(x)\right)\\
Q_{3}(x)&=x \cdot P_{1}(x, \alpha)-A(x) Q_{1}\left(x-P_{1}(v)\right)\\
Q_{4}(x)&=X \cdot P_{2}(x, \alpha)-B(x) \cdot Q_{2}\left(x-p_{2}(k)\right)
\end{align*}

Based on the relation and using a Lagrange polynomial, I get that the following equations would hold
\begin{align*}
0 &= M_{1} \cdot h^{\prime \prime}(x) \cdot h(x) +\sum_{k \in K} p_{1}(k+1)+M_{2} p_{2}(k-1)^{\alpha} +h^{\prime}(k) \cdot h^{\prime \prime}(k) \\
&= \sum_{k \in K}\left(h^{\prime}(k) +h^{\prime \prime}(k) \cdot M_{1}\right) \cdot \Lambda(x, \alpha) \cdot h^{\prime}(k) +P_{1}(x, \alpha) \cdot\left(M_{2} \cdot h(k)+P_{1}\right)
\end{align*}

Then, I get a relationship between the $P_i$ polynomials: 
$$H^{\prime}(x) \cdot P_{1}(x, \alpha)=P_{2}(x, \alpha) \cdot M_{2}+P_{3}(x,1)$$

After algebraic manipulations, I get 
\begin{align*}
P_{3}(x, 1)&=M_{2} \cdot P_{2}(x, \alpha)+ P_{1}(\alpha, y) \\
&=\sum_{k \in K}\left(h_{1}(K) \cdot v^{\prime}-h_{2}(K) \cdot u^{\prime}\right) \cdot Q(x)+P_{4}(y)
\end{align*}

Finally, I get that the prover's task to verify the polynomials $f(x),h(x)$ becomes
\begin{align*}
\sum_{k \in R}\left(h^{\prime}(k+m) -m \cdot h(k)\right) \cdot M_{2} + \Lambda(x, \alpha) \cdot h^{\prime \prime}(k-m) \cdot\left(P_{2}(x \alpha)-n\right) = 0
\end{align*}

\vspace{1mm}

\subsection{Online Phase}

\vspace{-5mm}

\subsubsection{Round 1}

During the first round of the online phase, the prover commits to certain polynomials that represent the computation or the statement being proved. 

I sampled six scalars $\left(a_{1}, \ldots, a_{6}\right)$ and two masking polynomials $b_1 \in F_{\alpha-a_{1}-a_{3}-1}[x]$ and $b_{2}(x) \in F_{\alpha-a_{2}-a_{4}-1}[x]$. I then obtain the following polynomials: 
\begin{align*}
& \hat{r}(x)=a_{1} Q(x)+a_{2} \sum_{i=1}^{\alpha} p_{2}(i) \cdot\left(x^{-i}-f(i)\right) \\
& \hat{s}(x)=a_{3}\left(h^{\prime}(x)-\sum_{k \in K} p_{1}(k) \cdot x^{\alpha}\right)+a_{4} q\left(m+n^{k}\right) \\
& \hat{t}(x)=h^{\prime \prime}\left(x-a_{4}\right)+\left(\hat{r}(x) \cdot \hat{s}(x) \cdot a_{6}\right)
\end{align*}

\subsubsection{Round 2}

This round involves the queries based on the commitment polynomials  received in the first round. The verifier is asking for more information to be able to verify the claim without gaining any knowledge about the actual content of the claim.

The encoder $R$ samples $\tau, \epsilon, \phi$ and I obtain the following polynomial checks from the verifier:
\begin{align*}
&\hat{g}(x)=b_{1}\left(a_{1} x^{2}-a_{3} x+a_{5}\right) \cdot \hat{r}(x+\varepsilon)+p_{1}(x, \tau) \\
&\hat{h}(x)=b_{2}\left(a_{2} x^{3}-a_{4} x^{2}+a_{6}\right) \cdot p_{2}(\varepsilon, x)-\alpha \hat{g}(x) \\
&\hat{f}(x)=(\hat{g}(x)+\alpha \cdot \hat{h}(x)
\end{align*}

Moreover, I get the following relationships: 
\begin{align*}
b_{1}^{\prime}(x)=b_{1}(x) \cdot A(x)+\sum_{k \in K} \hat{t}(k) \cdot p_{1}(x+\tau) \\
b_{2}^{\prime}(x)=b_{2}(x) \cdot(B(x)-Q(x-\not)) \bmod n
\end{align*}

\subsubsection{Round 3}

The prover responds to the verifier's queries by providing the requested evaluations of the committed polynomials, along with proofs that these evaluations are consistent with the original commitments. The prover aims to show that 
$$\sum_{m \in M} \hat{f}(m) = P_1(x,\tau,\epsilon)$$
for a set of matrices $M$ and random point $x \in \mathbb{F}$ sent by the verifier. 

$\hat{p}(x)$ is defined as follows: 
\begin{align*}
\hat{p}(x)=\sum_{m \in \mathbb{M}}\left(h^{\prime}(m)+\alpha \cdot h^{\prime \prime} (m)\right) \cdot \Delta_{\text {row }(m)}(x) \cdot \Delta_{\text {col }(m)}(x) \cdot \Delta_{m}^{\mathbb{M}}(X)
\end{align*}

From the equation above, I get that for any $m \in \mathbb{M}$
$$
\hat{p}(m) = \left(h^{\prime}(m)+\epsilon \cdot h^{\prime \prime} (m)\right) \cdot \Delta_{\text {row }(m)}^{\mathbb{H}}(x) \cdot \Delta_{\text {col }(m)}^{\mathbb{H}}(x)
$$

Because $\hat{p}(m)$ depends on the polynomials generated by the encoder, it is implicitly known by the verifier. Now, I compute $r_1$ and $r_3$ using the above equation. Let 
$$
r_1(x) = \frac{\tau \lvert M \rvert r_3(x)}{x^{\sigma}}
$$

Now, the equation becomes that for all $m \in \mathbb{M}$
\begin{align*}
n^{2} &\cdot r_1(m) \cdot(x-Q_4(M_1)) \cdot(y-Q_3(m))\\
&=\left(h^{\prime}(m)+\alpha \cdot h^{\prime \prime} (m)\right) \cdot Q_4(M_2) \cdot Q_3(m) \cdot \Delta_{\text {row }(m)}^{\mathbb{H}}(x) \cdot \Delta_{\text {col }(m)}^{\mathbb{H}}(x)
\end{align*}

I use algebraic manipulations on this equation to get that for all $m \in \mathbb{M}$: 
\begin{align*}
&\left(m \cdot r_3(m)+\frac{\epsilon}{|M_1|}\right) \cdot n^{2} \cdot(x y+Q_2(m)-x \cdot Q_3(m)-y \cdot Q_4(m)) \\
&-\left(Q_1 (M_1)+\tau \cdot Q_1(M_2)\right) \cdot \Delta_{\text {row }(m)}^{\mathbb{H}}(x) \cdot \Delta_{\text {col }(m)}^{\mathbb{H}}(x)=0
\end{align*}

Now, I define an auxiliary polynomial $T(x)$
\begin{align*}
T(X) & = \lvert M \rvert \cdot n^{4} \cdot(x^4 +Q_2(X)-x^3 \cdot Q_3(X)-x^2 \cdot Q_4(X))\\
&-(Q_1(X) +\alpha \cdot C_s(x)) \cdot Z_{\mathbb{H}}(x) \cdot Z_{\mathbb{H}}(y)\\
&+r_3(X) \cdot n^{2} \cdot(x y \cdot X+Q_2^{\prime} (X)-x^3 \cdot Q_3^{\prime}(X) - x^2 \cdot Q_4^{\prime}(X))
\end{align*}

Thus, the prover can now compute $r_2$, which it sends to the verifier.
$$
r_2(X):=\frac{t(X)}{Z_{\mathbb{M}}(X)} \in \mathbb{F}_{\leq|\mathbb{M}|-2}[X]
$$

\subsection{Decision phase}

The decision phase is where the verifier processes the prover's responses and decides whether the proof is valid. This decision is based on whether the responses from the prover satisfy all the conditions set out by the verifier's challenges. The verifier first runs the \textbf{VerifyPoly} and \textbf{VerifyEval} algorithms of LUMEN's PCS, and then it checks the following relationships between polynomials and returns $\{0,1\}$ to show if it accepts or rejects the generation of these arguments of knowledge: 

\begin{align*}
    x\hat{s}(x) &+ \left(\hat{a}(x) \cdot \hat{r}(x) + \sum_{h \in \mathbb{H}} x^{\alpha} \cdot \Delta_{h}(x)\right) \cdot \left(\Lambda_{\mathbb{H}}(x, 1) + \left(\hat{b}(x ) \cdot Z_{\mathbb{H}}(x ) + 2\hat{t}(x) \right) \cdot Q(\sigma) \right) \\
    &+ \left(\hat{b}(x) \cdot \Delta_h(x) + 1\right) \cdot \alpha \Lambda_{\mathbb{H}}(\tau, \epsilon) - Q_3^{\prime}(x)\hat{p}(x) \stackrel{?}{=} 0
\end{align*}

\begin{align*}
\sigma &\cdot (\tau^2 + 3\epsilon + 2) \cdot(x^3 y+Q_2(X) +x^3 \cdot Q_3(X) + x \cdot Q_4(X)) \\
&-r_3(X) \cdot n^{2} \cdot\left(x^2 \cdot Q_2^{\prime}(X)+x \cdot Q_3^{\prime}(X) + y \cdot Q_4^{\prime}(X)\right) \\
&-\left(Q_1(X) - \epsilon \cdot Q_1(X)\right) \cdot Z_{\mathbb{H}}(x) \cdot \Delta_{\text {row }(m)}^{\mathbb{H}}(x) + r_2(x) \Delta_{\text {col }(m)}^{\mathbb{H}}(x) \stackrel{?}{=} 0
\end{align*}

\section{Security Analysis}\label{6}

\subsection{Completeness}

\begin{thm}
The polynomial commitments compiled by LUMEN's PIOPs are complete. 
\end{thm}

\begin{proof}

To prove completeness, I wish to show that for a valid commitment $c(x)$ generated by any honest prover $\mathcal{P}^{\prime}$ using $(p_1(x),p_2(x),u^{\prime},v^{\prime})$, 
\begin{equation*}
    Pr[(\mathcal{P}^{\prime},\mathcal{V})[c(x)] = 1] = 1 - \operatorname{negl}(x)
\end{equation*}

During the eval of each recursive step, $\mathcal{P}^{\prime}$ maps the coefficients of polynomials $p_1(x)$ and $p_2(x)$ to $h_1(x)$ and $h_2(x)$. Thus, for any prime $p$,
$$p_1\left(x_{1}, \ldots, x_{\mu}\right) \bmod p=h_1\left(x_{1}, \ldots, x_{\mu}\right)=y_1$$
$$p_2\left(x_{1}, \ldots, x_{\mu}\right) \bmod p=h_2\left(x_{1}, \ldots, x_{\mu}\right)=y_2$$

As a result, I obtained the following relationships that demonstrate the completeness of the parameters during the setup phase of $\mathcal{P}^{\prime}$
$$y_1 = \frac{u^{\prime}-h_1(x)}{\alpha} \text{ }; \text{ } y_2 = \frac{v^{\prime} - h_2(x)}{\alpha} \text{ }; \text{ } q(x) = y_1 + \sqrt{\sum_{v_i \in v^{\prime}} v_i^2} y_2$$

Moreover, in round $k$ of the Eval protocol, $\mathcal{P}^{\prime}$ can compute polynomials $H_1$ and $H_2$ by aggregating $h_1$ and $h_2$, respectively, to obtain 
\begin{align*}
f\left(X_{1}, \ldots, X_{k}\right) = H_1\left(X_{k}, \ldots, X_{k-1}\right)+X_{k} H_2\left(X_{1}, \ldots, X_{k-1}\right)
\end{align*}

Consequently, 
$$v^{{\prime}^{H_1(x)+\frac{q^d - 1}{2} H_2(X)}}= v$$

Thus, for $p^{\prime}=H_1+ \lvert v^{\prime} \rvert \cdot H_2 \in \mathbb{Z}\left(2^{\alpha} \cdot b\right)$, a polynomial with coefficients bounded by $\left(2^{d}-1\right) \cdot b+b=2^d b$, I have
\begin{align*}
p^{\prime}\left(x_{1}, \ldots, x_{i-1}\right) \bmod p &=H_1\left(x_{1}, \ldots, x_{i-1}\right)+ \lvert v^{\prime} \rvert \cdot H_2\left(x_{1}, \ldots, z_{i-1}\right) \bmod p \\
&=y_1+ \lvert v^{\prime} \rvert \cdot y_2 \bmod p \\
&= q(\alpha)
\end{align*}

In the final round of the protocol, $\mathcal{P}^{\prime}$ checks whether $\alpha \cdot q \cdot c(x) \equiv p_1(x) \mod \mathbb{Z}[X]/(x^{\alpha} - 1)$, which is true by construction. Thus, a valid commitment is accepted, and the arguments of knowledge are complete. 

\end{proof}

\subsection{Soundness}
\begin{thm}
LUMEN's PCS compiled with its PIOP produce sound arguments of knowledge. 
\end{thm}

\begin{proof}

First, the arguments of knowledge produced by LUMEN's PCS already have witness-extended emulation, which is a stronger version of knowledge-soundness. Thus, I only need to consider the PIOP protocol. 

To prove knowledge soundness, I wish to show that when adversary $\mathcal{A}$ generates an invalid commitment $c^{\prime}(x)$ with prover $\mathcal{P}^{\prime}$ that does not have the true parameters, I have
\begin{equation*}
    Pr[(\mathcal{P}^{\prime},\mathcal{V})[c(x)] = 1] = \operatorname{negl}(x)
\end{equation*}

I define extractor $\mathcal{E}$, which is simply the algorithm that runs the prover $\mathcal{P}^{*}$ for the first round, and obtains $\hat{r}(x), \hat{s}(x), \hat{t}(x)$, which determines the polynomial checks in the decision phase. In the field $\mathbb{F} = \mathbb{Z}[x]/(x^{\alpha} - 1)$,
\begin{align*}
f(x)&=\sum_{h \in \mathbb{F}}(\hat{r}(h)+ M_1 \cdot \hat{s}(h)) \cdot \Lambda_{\mathbb{H}}(X, h)+\hat{r}(h) \cdot \hat{s}(h) \cdot P_3(x, h) \neq 0  
\end{align*}

I aim to demonstrate that the bounds on $f(x)$ ensure that the invalid $f(x)$ does not pass the polynomial checks during the decision phase. I denote all parameters and variables generated by the true prover with subscript $1$, and all invalid parameters and variables generated by $\mathcal{P}^{\prime}$ with subscript $2$. 

To achieve my goal, I set the following variables to consider the bounds of my parameters in the polynomial check. Let $\Chi = \lambda \epsilon^2$, $\Psi = \log{\lambda} + \epsilon$, and $\Omega = 2^{\alpha} \cdot (\epsilon - \lambda)$. Then, the bound on $2^{\alpha} - 1$ is 
\begin{align*}
    2(\lambda+1)+\log \left( \alpha_1 \right)
    &=2\left(\lambda+\Chi+1\right)+\Psi^2+2\Omega \\
    &\leq \frac{1}{2}\left(\log q-1+\Psi^2+2\Omega\right) \\
    &=\log \alpha_{2}
\end{align*}

Moreover, I bound the norm of this invalid commitment $c(x)$ to show that the prover and verifier will not accept $c(x)$. First, I know that for all $j \in \{ 1,2\}$, I obtain
$$h_1(x) + \alpha_j h_2(x) = Q_1(x^{-j}) + Q_2(x^{2^{\alpha}-j})$$
$$y_1+ \lvert v^{\prime} \rvert y_2 = P_1(M_1,y) + (2^{d} - 1 )P_2(M_2,y) \mod p$$

Assume, for the sake of contradiction, that opening $q_2,p_{2_2}$ is indeed a valid opening to commitment $c_2$. Then, I get that 
$$\forall_{j \in\{1,2\}} q_{j} \cdot\left(c_1+\alpha_{j} c_{2}\right)=f_{j}(v)$$
$$h_{j}(v)=q_{j}^{-1} f_{j}(z)=y_{1}+ \alpha_j y_{2} \bmod p$$  

Let $L: \mathbb{Z}[X]/(x^{\alpha} - 1) \rightarrow \mathbb{Z}^2$. Then, I get that if $\alpha_{1} \neq \alpha_{2}$,
$$
\begin{aligned}
& {\left[\begin{array}{cc}
q_{1} & 0 \\
0 & q_{2}
\end{array}\right]\left[\begin{array}{ll}
1 & \alpha_{1} \\
1 & \alpha_{2}
\end{array}\right]\left[\begin{array}{l}
c_{1} \\
c_{2}
\end{array}\right]=L_{q}\left(\left[\begin{array}{l}
f_{1} \\
f_{2}
\end{array}\right]\right) \cdot G} \\
& \left(\alpha_{2}-\alpha_{1}\right) N_{1} N_{2}\left[\begin{array}{l}
c_{1} \\
c_{2}
\end{array}\right]=L\left(\left[\begin{array}{cc}
\alpha_{2} q_{2} & -\alpha_{2} q_{1} \\
-q_{2} & q_{1}
\end{array}\right]\left[\begin{array}{l}
f_{1} \\
f_{2}
\end{array}\right]\right) \cdot \mathrm{G}
\end{aligned}
$$

This implies that $q \cdot c_{1}=f_{1}(v) \cdot G, q \cdot c_{2}=f_{2}(\vec{v}) \cdot G, q \cdot y_{1}=f_{1}(\vec{v}) \bmod p$, and $q \cdot y_{2}=f_{2}(\vec{v}) \bmod p$. Furthermore, $\forall_{j}|| f_{j} \|_{\infty} \leq B_{i}$ and $\forall_{j}\left|q_{j}\right| \leq B_{i}$, so $|q| \leq 2^{\lambda+1} D_{i}^{2}$ and $\|f\|_{\infty} \leq 2^{\epsilon+1} B_{i} D_{i}$.

Finally, combining all of these, I obtain the following probability for the prover to accept an invalid $c(x)$ such that $h_{\mu}=\frac{1}{N} \cdot f\left(\alpha_{\mu}, \ldots, \alpha_{i+1}\right)$. I split the probability into two cases: 

\textbf{Case 1: } $Q>D_{i}$

By the Multilinear Composite Schwartz-Zippel Lemma [5], I get that the probability that $h_{\mu} \in \mathbb{Z}$ is
\begin{align*}
\operatorname{Pr}\left[f\left(\alpha_{\mu}, \ldots, \alpha_{i+1}\right) \equiv 0 \bmod q \right]\leq \frac{\lambda - 2\alpha +3}{2^{\epsilon}}
\end{align*}

\textbf{Case 2: } $Q\leq D_{i} \wedge\|f\|_{\infty}>B_{i}$

If $\left|h_{\mu}\right| \leq B_{\mu}$ then $\left|f\left(\alpha_{\mu}, \ldots, \alpha_{i+1}\right)\right| \leq q \cdot B_{\mu} \leq \Chi_{i-1} \cdot B_{\mu}$. Since $\|f\|_{\infty}>B_{i}$ and $\log B_{i}= \Chi_i +\Omega_i+\log B_{\mu}$, I get that
$$\log \left(\Chi_{i-1} \cdot B_{\mu}\right) \leq \log B_{i}-\Omega_{i-1}<\log \|f\|_{\infty}-\Omega_{i-1}$$

Therefore, the probability would be
\begin{align*}
    \operatorname{Pr}[h_{\mu}\leq B_{mu}] 
    &\leq \operatorname{Pr}[\left|f(\alpha_{\mu},\ldots,\alpha_{i+1})\right| \leq D_i \cdot B_{\mu}] \\
    &\leq \operatorname{Pr}[\left|f(\alpha_{\mu},\ldots,\alpha_{i+1})\right| \\ 
    &\leq \frac{1}{2^{\Omega_i}} \cdot \|f\|_{\infty}] \\
    &\leq \frac{2(\lambda - \epsilon)}{\alpha}
\end{align*}

Combining the two cases, I find that the probability of the prover accepting an invalid commitment is negligible as long as the requirements for the parameters are met. 

\end{proof}

\subsection{Zero-Knowledge}

\begin{thm}
LUMEN's PCS and PIOP protocols are perfect zero-knowledge and compile zero-knowledge arguments. 
\end{thm}

\begin{proof}

To prove zero-knowledge, I wish to show that the transcript generated by Simulator $\mathcal{S}$ is based on queries $(G,u,v,\lambda,p_1(x),r^{\prime}(x),s^{\prime}(x),t^{\prime}(x))$ and is indistinguishable from a real transcript between a prover and a verifier. Based on these queries, simulator $\mathcal{S}$ computes $m^{\prime} {\leftarrow}[0, k-1]$, $d {\leftarrow}$ $\operatorname{Primes}(\lambda)$, $k_{1}^{\prime}, \ldots, k_{n}^{\prime}, \stackrel{\$}{\leftarrow}[0, \alpha-1], \text{ and } v_{1}^{\prime}, \ldots, v_{n}^{\prime} \stackrel{\$}{\leftarrow}[0, d-1]$. 

Then, for every query $(r^{\prime},s^{\prime},t^{\prime})$, the simulator samples $\hat{t}(x) \leftarrow_{\&} \mathbb{Z}[x]$ and computes $t_{i}$ as follows: 

$$
\begin{aligned}
t_{\alpha} & \leftarrow t_{1, \epsilon} \cdot \Delta_{\mathbb{L}}(\epsilon)+\sum_{h \in \mathbb{L}} x^{\alpha} \cdot \mathcal{L}_{h}^{\mathbb{H}}(\epsilon) \\
t_{d} & \leftarrow t_{2, \epsilon} \cdot \boldsymbol{Z}_{\mathbb{L}}(\epsilon)+1 \\
t_{m} & \leftarrow\left(t_{a, \epsilon}+\epsilon \cdot t_{b, \epsilon}\right) \cdot \Lambda_{\mathbb{H}}(x, \epsilon)+t_{u, \epsilon} \cdot t_{v, \epsilon} \cdot P_4(w, \epsilon, \alpha) \\
t_{n} & \leftarrow \frac{t_{p, \epsilon}+t_{a, \epsilon}-\alpha \cdot t_{b, \epsilon}}{\Delta_{\mathbb{H}}(\epsilon)}
\end{aligned}
$$

First, I want to show that for any $g_i$ sampled from group $G$ of unknown order, $g_i^{k_i}$ is statistically indistinguishable from group $u$. Because the distribution of $g_i^{k_i}$ on $G$ is the same as that of $k_i \bmod \left|G^{\frac{1}{2} \alpha }\right|$, the probability that $k_i = p$ for some exponent $p$ of the generator $g \in u$ is the same as the probability that $k_i$ falls within the interval $[nd, n(d+1) - 1]$. Thus, I get that 
$$\operatorname{Pr}\left[k_{i}=n\right]= \frac{d}{2\alpha+1}$$ 
holds when $n d \geq m x_{i}-\alpha$ and $(n+1) d-1 \leq m x_{i}+\alpha$. 

Let $X$ be the set of points that satisfies the above equation. These points are in the form of $\left\lfloor\frac{m x_{i}-\alpha}{d}\right\rfloor$ and $\left\lfloor\frac{m x_{i}+\alpha}{d}\right\rfloor$, which are the endpoints of inequalities. Let $Z$ be a set of points with $\operatorname{Pr}\left[p_{i}=z\right]=\zeta /(2 D+1)$. The discrepancy of $p_{i}$ from a uniform random variable $V_{Z}$ across $Z$ is highest when the count of feasible $q_{i}$ mapping to $F_{1}$ and $F_{2}$ is both $\zeta-1$, that is, $e y_{i}-D=1 \bmod \zeta$ and $e y_{i}+D=\zeta-2$ $\bmod \zeta$. In this case, $p_{i}$ is one of the two endpoints outside $Z$ with likelihood $\frac{2(\zeta-1)}{2 D+1}$. As $|Z|=\frac{2 D+3}{\zeta}-3$, the statistical discrepancy of $p_{i}$ from $V_{Z}$ is at most 
\begin{align*}
&\frac{1}{2}\left(|Z|\left(\frac{1}{|Z|}-\frac{\zeta}{2 D+1}\right)+\frac{2(\zeta-1)}{2 D+1}\right)\\
&=\frac{5 \zeta-4}{2(2 D+1)} \leq \frac{2^{\delta+1}}{D}
\end{align*}

$V_{Z} \bmod \left|G_{i}\right|$ displays a discrepancy at most $$\frac{\left|G_{i}\right|}{|Z|} \leq \frac{2^{\delta}|G|}{2 D+3-3 \cdot 2^{\delta}}<\frac{2^{\delta-1}|G|}{D+1-2^{\delta}}$$

According to the triangle inequality, the discrepancy in $p_{i} \bmod \left|G_{i}\right|$ from the uniform distribution is at most $$\frac{2^{\delta+1}}{D}+\frac{2^{\delta-1}|G|}{D+1-2^{\delta}}$$

This also bounds the discrepancy in $g_{i}^{p_{i}}$ from the uniform distribution across $G_{i}$. The emulated value $g_{i}^{p_{i}^{\prime}}$ exhibits a maximum discrepancy of at most $|G| / D$ from the uniform distribution across $G_{i}$. According to the triangle inequality, the discrepancy between $g_{i}^{p_{i}}$ and $g_{i}^{p_{i}^{\prime}}$ is at most:

\begin{align*}
\frac{2^{\delta+1}}{D}+\frac{2^{\delta-1}|G|}{D+1-2^{\delta}}+\frac{|G|}{D} &< \frac{2^{\delta-1}|G|+|G|+2^{\delta+1}}{D+1-2^{\delta}} = \frac{\left(2^{\delta-1}+1\right)|G|+2^{\delta+1}}{D+1-2^{\delta}} \\
&\leq \frac{1}{n 2^{\epsilon_{a}+1}},
\end{align*}

if $$D \geq n2^{\epsilon_{a}+1}\left(2^{\delta-1}+1\right)|G|+n 2^{\epsilon_{a}+\delta+2}+2^{\delta}-1$$ for some discrepancy parameter $\epsilon_{a}$.

Finally, I analyzed the distribution of $w_i$ where $w_i = g_i^{k_i} \cdot v_i$ and the conditional distribution of $k_{i} \mid v_i$. Consider the conditional distribution of $p_{i} \mid y_{i}$. Note that $k_{i}=z$ when $\left(k_{i}-y_{i}\right) / \zeta=z$. An argument similar to that above is repeated to bound the distribution of $p_{i}$ from the uniform distribution. I can compare the joint distribution $X_{i}=\left(g_{i}^{p_{i}}, y_{i}\right)$ with the emulated distribution $Y_{i}=\left(g_{i}^{p_{i}^{\prime}}, y_{i}^{\prime}\right)$. With $\epsilon_{1}=\frac{1}{|Z|}-\frac{\zeta}{2 D+1}$ and $\epsilon_{2}=0$, the discrepancy between these joint distributions is at most:
\begin{align*}
\frac{1}{n 2^{\epsilon_{a}+1}}+\frac{\zeta}{2 D+1}+\epsilon_{1} \zeta = \frac{1}{n 2^{\epsilon_{a}+1}}+\frac{\zeta^{2}}{2 D+3-3 \zeta}+\frac{\zeta(1-\zeta)}{2 D+1}
\end{align*}

Furthermore, as each $X_{i}$ is independent of $X_{j}$ for $i \neq j$, I bound the discrepancy between the joint distributions $\left(g_{1}^{p_{1}}, \ldots, g_{n}^{p_{n}}, y_{1}, \ldots, y_{n}\right)$ and $\left(g_{1}^{p_{1}^{\prime}}, \ldots, g_{n}^{p_{n}^{\prime}}, y_{1}^{\prime}, \ldots, y_{n}^{\prime}\right)$ by summing the individual discrepancies between each $X_{i}$ for $x \in\left[0, \tilde{z} \cdot 2^{80}\right]$ and $Y_{i}$, as follows:
\begin{align*}
\frac{1}{2^{\epsilon_{a}+1}}+\frac{n \zeta^{2}}{2 D+3-3 \zeta}+\frac{n \zeta(1-\zeta)}{2 D+1} &= \frac{1}{2^{\epsilon_{a}+1}} +\frac{n \zeta\left(2 D+3-5 \zeta-3 \zeta^{2}\right)}{(2 D+3-3 \zeta)(2 D+1)} \\
& < \frac{1}{2^{\epsilon_{a}}+1}+\frac{n \zeta}{2 D+1}<\frac{1}{2^{\epsilon_{a}}},
\end{align*}

where the last equality holds if $D>n 2^{\epsilon_{a}+\delta}-1$. Finally, this also bounds the discrepancy between $(P, \boldsymbol{y})$ and $\left(P^{\prime}, \boldsymbol{y}^{\prime}\right)$, where $P=\prod_{i} g_{i}^{p_{i}}$ and $P^{\prime}=\prod_{i} g_{i}^{p_{i}^{\prime}}$. Combining the two conditions for $D$, I obtain $D \geq n 2^{\delta+\epsilon_{a}+1}|G|$.

Finally, I consider the $t_{i}$'s generated by the simulator. By its construction, I see that it is generated by the same polynomials that the prover and verifier use; therefore the $t_{i}$'s in the simulator transcript will be identical to the real transcript. 

\end{proof}

\section{Implementation}\label{7}

I first apply the Fiat–Shamir heuristic to the PIOP protocol to transform the transparent interactive arguments of knowledge into transparent and non-interactive arguments of knowledge (i.e. transparent zk-SNARKs). Then, I implemented LUMEN by coding $8000$ lines of Rust and Python code [27]. I ran my implementation under different conditions such as different Ethereum transaction size to ensure that my efficiency measurements (proof size, prove computation time, and verification time) are accurate and representative of LUMEN. 

\subsection{Results}

I chose four state-of-the-art methods to compare LUMEN to: two transparent SNARKs (zk-STARK and DARK) and two non-transparent SNARKs (Plonk and Sonic). We first compare the proof size of LUMEN with that of the other algorithms. The transparent SNARKs are in green while the non-transparent SNARKs are in blue. 

\begin{figure}[H]
    \centering
    \includegraphics[width=0.7\textwidth]{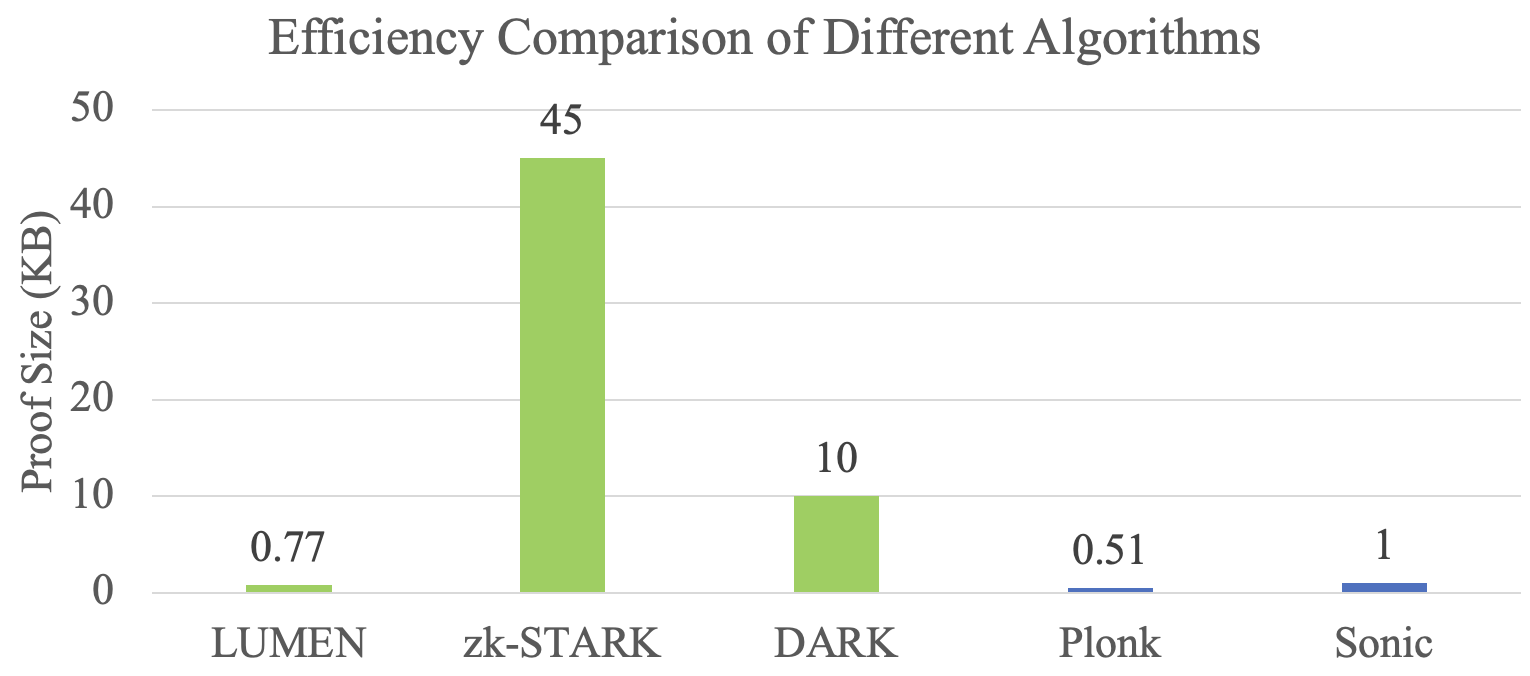}
    \caption{Proof Size Comparison of Different Algorithms \\ \textbf{Source}: Myself}
\end{figure}

As seen in the table on the top of the next page, LUMEN is around $58$ times more efficient than zk-STARK and $13$ times more efficient than DARK. Moreover, LUMEN is on par with the efficiency of the non-transparent zk-SNARKs: at around $1$KB. Thus, I demonstrated that LUMEN's transparency does not significantly affect its efficiency and that LUMEN is a promising solution to enhance ZKR's security with its transparent zk-SNARKs. 

When comparing prover time and verifier time, we consider the parameter $n$, which is the number of gates. In our comparison, we let the number of gates be constant and $2^{20}$ for a security level of $120$-bit. Thus, we get the following comparisons. 

\begin{figure}[H]
    \centering
    \includegraphics[width=0.7\textwidth]{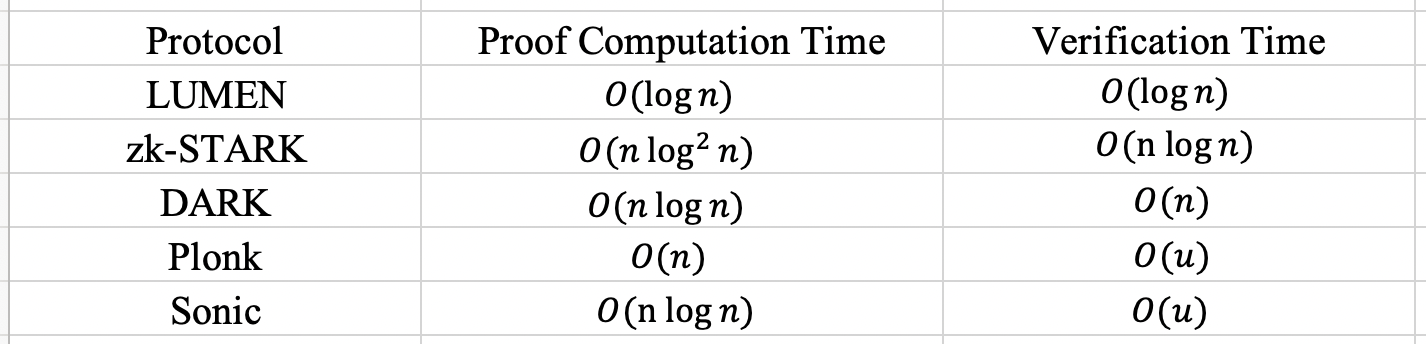}
    \captionsetup{justification=centering}
    \caption{Proof Computation Time and Verification Time Comparisons \\ \textbf{Source}: Myself}
\end{figure}

As seen on the table, it is clear that LUMEN clearly surpasses the transparent zk-SNARKs with its small proof computation and verification time, and those times are comparable to those of Plonk and Sonic, which are non-transparent zk-SNARKs. As demonstrated above, LUMEN's efficiency well exceeds the efficiency of transparent zk-SNARKs and is similar to the efficiency of non-transparent zk-SNARKs. 

\section{Conclusion}\label{8}

This research proposes LUMEN, a novel set of algorithms that generate transparent zk-SNARKs that enhance Ethereum's security without sacrificing its efficiency. LUMEN is a promising solution that could be implemented into ZKR and Ethereum, benefiting the market worth over 220 billion USD. This security improvement will aid Ethereum's path to more scalability and has the potential to impact the security standards of the entire cryptocurrency market. 

LUMEN's proposed zero-knowledge system has several advantages, including its transparency, efficiency, practicality, and easy implementation because of its PCS's witness-extended emulation. On the other hand, there are also a few areas of future work to be done. Researchers could further modify LUMEN's algorithms to improve Ethereum's security or efficiency, generalize LUMEN's algorithms to other cryptocurrencies, and utilize the innovative techniques within LUMEN in other algorithms or areas of cryptography to enhance security or efficiency.

\end{document}